\documentclass[11pt,letterpaper]{article}
\usepackage{amsfonts,color,ifpdf,latexsym,graphicx,url,amssymb,enumitem,amsthm,wrapfig}
\title{Rigid Foldability is NP-Hard}
\author{Hugo Akitaya \and Erik D. Demaine \and Takashi Horiyama \and Thomas C. Hull \and Jason S. Ku \and Tomohiro Tachi}
\date{}

\usepackage{amsmath}
\usepackage{bm}

\def\vec{\mathbf} %overloading \vec
\newcommand{\degrees}{^\circ}

\graphicspath{
 {./images/}
}

\newif\ifabstract
\abstracttrue
\abstractfalse
\newif\iffull
\ifabstract \fullfalse \else \fulltrue \fi

\newif\ifanalysis
\analysisfalse
\usepackage
  [breaklinks,bookmarks,bookmarksnumbered,bookmarksopen,bookmarksopenlevel=2]
  {hyperref}
{\makeatletter \hypersetup{pdftitle={\@title}}}
\advance \topmargin by -\headheight
\advance \topmargin by -\headsep
\evensidemargin \oddsidemargin
{\makeatletter \gdef\fps@figure{!tbp}}

\let\realbfseries=\bfseries
\def\bfseries{\realbfseries\boldmath}

\newtheorem{theorem}{Theorem}
\newtheorem{lemma}[theorem]{Lemma}
\newtheorem{corollary}[theorem]{Corollary}
\newtheorem{remark}[theorem]{Remark}
\theoremstyle{definition}
\newtheorem{definition}[theorem]{Definition}
\let\epsilon=\varepsilon

\begin{document}
\maketitle
\section{Introduction}
Intuitively, a rigid origami is a developable surface composed of planar rigid facets and rotational hinges. 
Determining the flexibility of such system is called \emph{rigid foldability}, i.e., the problem of judging if a given planar crease pattern can flex finitely to an intermediate state without deforming its facets. 
In practical sense, rigid-foldable origami is a mechanism whose flexibility does not rely on irreversible deformation of materials; it only relies on rotations around predefined hinges.  
It therefore has enormous applications in different fields and at different scales, e.g., self-folding mechanisms of microscopic material \cite{SilverbergScience}, foldable packaging \cite{Dai15}, transformable adaptive architecture with thick panels \cite{Tachi2}, deployable space structures \cite{Schenk}, and so on.
Despite its usefulness, we do not have many general results about rigid foldability.  We can determine easily when a single vertex is rigid-foldable \cite{ACDEHKLT}, but most other results  are on determining if specific examples are rigidly foldable, e.g., \cite{Tachi:2010:bdffpq}.
This is in contrast to to the flat-foldability, which has been well investigated; single-vertex necessary conditions are known \cite{Hull}, as are linear time algorithms for single vertex \cite{GFALOP} and the NP-hardness for general crease pattern \cite{bern-hayes}, and even for box pleating~\cite{box-pleating}.

\begin{wrapfigure}[12]{r}[0pt]{0.2\linewidth}
	\includegraphics[width=\linewidth]{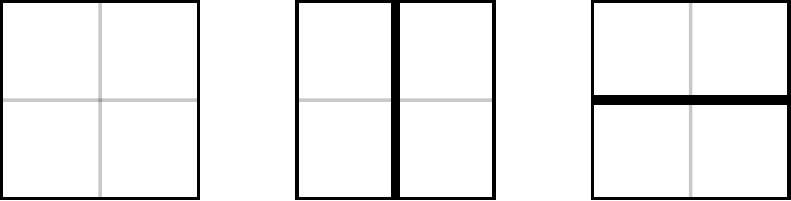}
	\centering
	\caption{Four lines meeting perpendicularly can rigidly fold from the flat state with optional creases, but not using all creases.}
	\label{fig:overlay}
\end{wrapfigure}

Our objective is to characterize rigid foldability in general.
In this paper, we focus on rigid foldability from a planar sheet of paper,
and we introduce the following types of rigid foldability:
\begin{description}
\item [Rigid foldability (with optional creases)]
A \emph{continuous rigid folding motion} is a continuous transformation of a sheet of paper from one state to the other such that each face, i.e., a region separated by creases, is rigid. 
We say a crease pattern \emph{rigidly foldable} (with optional creases) if there is a continuous rigid folding motion from the initial state that is not a rigid body motion.
\item [Rigid foldability using all creases]
We say a crease pattern is \emph{rigidly foldable using all creases} if there is a continuous rigid folding motion where the dihedral angle at every crease strictly increases or decreases.
\end{description}

For example, a single vertex with four perpendicular lines Figure~\ref{fig:overlay} can rigidly fold either vertically or horizontally, but cannot fold using all four lines.
So, it is rigidly foldable with optional creases, but not using all creases.
Single-vertex rigid foldability for a given assignment is studied in \cite{ACDEHKLT}; this provides a simple necessary and sufficient condition that there must be a tripod or X pattern incident to the vertex. 
However, a mere intersection of such single vertex condition does not lead to useful condition of multi-vertex rigid-foldability, since the combination of multiple vertex forms a globally constrained system of folding speeds.

In this paper, we show that deciding rigid foldability of a given crease pattern using all creases is weakly NP-hard by a reduction from Partition, and that deciding rigid foldability with optional creases is strongly NP-hard by a reduction from 1-in-3 SAT.
Unlike flat foldability of origami or flexibility of other kinematic linkages, whose complexity originates in the complexity of the layer ordering and possible self-intersection of the material,  rigid foldability from a planar state is hard even though there is no potential self-intersection.
In fact, the complexity comes from the combinatorial behavior of the different possible rigid folding configurations at each vertex.
The results underpin the fact that it is harder to fold from an unfolded sheet of paper than to unfold a folded state back to a plane,
frequently encountered problem when realizing folding-based systems such as self-folding matter and reconfigurable robots.

\subsection{Mathematical Model}

We model rigid folding with the concept of isometric foldings as introduced by Robertson \cite{Robertson}. 
\begin{definition}[Folding]
Given two manifolds, $M$ and $N$, with or without boundary, we call a function $f:M\rightarrow N$ an \emph{isometric folding} if $f$ maps finitely-piecewise geodesic curves in $M$ to finitely-piecewise geodesic curves in $N$, where both curves are parameterized with respect to arc-length.
It is not difficult to show that all isometric foldings under this definition must also be continuous.  One special case is where dim $M=$ dim $N$, which we call a \emph{flat folding}.  For any isometric folding $f$, the set $\Sigma(f)\subset M$ of points where $f$ is non-differentiable partitions $M$ into a finite cell complex (see \cite{Robertson}), which we call the \emph{crease pattern} of $f$. 
\end{definition} 

\begin{definition}[Rigid Origami Folded State]
Let $M\subset {\mathbb R}^2$  be a closed polygon which will model our material to be folded.  Then the crease pattern $\Sigma(f)$ will be a planar straight-line graph embedded on $M$.  If we also have that $N={\mathbb R}^3$ and the domain material $M$ does not cross itself in the image $f(M)$, then we call $f$ a \emph{rigid origami folded state}; this word choice is justified because $f(M)$ will be a continuous mapping of $M$ into ${\mathbb R}^3$ where the flatness of each 2-cell in $\Sigma(f)$ is preserved in $f(M)$. 
Each 1-cell of $\Sigma(f)$ is called a \emph{crease line}, and at each crease line $c_i$ we will refer to the \emph{folding angle} of $c_i$, denoted $\rho_i$, to be the signed angle by which the  two faces (2-cells) adjacent to $c_i$ deviates from a flat plane in the image $f(M)$.  We say that $c_i$ is a \emph{valley} (resp. \emph{mountain}) crease if $\rho_i>0$ (resp. $\rho_i<0$).  The information of which creases are mountains and which are valleys can be thought of as a mapping from the crease lines to the set $\{$M, V$, 0\}$ and is called the \emph{mountain-valley (MV) assignment} of the rigid origami folding.  The label ``0" is reserved for the case of \emph{optional creases} which have a folding angle of $\rho_i=0$.  Such creases are not technically part of $\Sigma(f)$, but later on we will want to include optional creases as part of the combinatorial structure in our complexity proofs.
\end{definition}

Let us denote by $R_{c_i}(\rho_i)$ the rotation matrix in ${\mathbb R}^3$ about a crease line $c_i\in\Sigma(f)$ by angle $\rho_i$ if we imagine $M$ to be embedded in the $xy$-plane of $\mathbb{R}^3$.  Now let $\gamma$ be any simple, closed, vertex-avoiding curve drawn on $M$ that crosses, in order, the crease lines $c_1,\ldots c_n$ in $\Sigma(f)$.  A necessary (but not sufficient, since the paper may cross itself) condition for $f$ to be a rigid origami folding is
\begin{equation}\label{eq1}
R_{c_1}(\rho_1)R_{c_2}(\rho_2)\cdots R_{c_n}(\rho_n)=I
\end{equation}
for all such curves $\gamma$ (see \cite{belcastro}).  Note, however, that when folding by a small amount from the flat, unfolded state (i.e., when $\rho_i=0$ for all $i$), there will be no chance of self-intersection. Under these conditions, which are the ones we will be using in this paper, Equation~\eqref{eq1} will be a necessary and sufficient condition for rigid foldability.

\begin{definition}[Rigid Folding Motion]
Equation~\eqref{eq1} can be used to compute relationships among the folding angles $\rho_i$, while rigid-bar and spherical kinematics \cite{Chiang} can be applied to compute the degrees of freedom of the system of equations made by the creases in $\Sigma(f)$.  
The result is a \emph{continuously parameterized family} of rigid origami folded states that describe the rigid folding process from one folded state to another folded state (an immersion in ${\mathbb R}^3$).   
We also call this continuous parameterized family a \emph{rigid folding motion}.
\end{definition}

\begin{definition}[Rigid Foldability]
We then say that a crease pattern $\Sigma(f)$ is \emph{rigidly foldable} between two isometric foldings $f:M\subset{\mathbb R}^2 \rightarrow {\mathbb R}^3$ and $g:M\subset{\mathbb R}^2 \rightarrow {\mathbb R}^3$ if there exists a rigid folding motion $f_{\vec{v}}:M\rightarrow{\mathbb R}^3$ between between $f$ and $g$, where $\vec{v}$ is the parameter vector, $\vec{v}=\vec{0}$ is the unfolded state, and $\vec{v}=\langle \rho_1, \ldots \rho_n\rangle$ is the set of parameters that achieve $f$.
In other words, there exists a homotopy $H_t:M\times [0,1] \rightarrow f(M)$ such that $H_t$ is a rigid origami folded state with $\Sigma(H_t)=\Sigma(f)$ for all $t\in[0,1]$, and $H_t(0) = f$ and $H_t(1) = g$.
\end{definition}

\subsection{Computational Model}

To formally state the computational problems we consider, we need to deal
with the limitation of digital computers which cannot exactly represent real
numbers.  Thus, we need to specify a finite digital representation of crease
patterns, and we may need to tolerate some error in deciding whether a
rigid motion is actually valid.  For now, we define exact versions of the
problems (with no error), and leave approximate versions for
Definition~\ref{approximate definition}.

\begin{definition}[Rigid Foldability from Flat State]
\label{exact definition}
We define two decision problems.  In both cases, we are given
a straight-line planar graph drawing $G$ on a polygon $M\subset{\mathbb R}^2$,
where all vertex coordinates are specified by rational numbers.

\begin{enumerate}
\item
\emph{Rigid foldability using all creases}:  Is there a rigid folding motion from the trivial (unfolded) folding $f=\textrm{identity}$ to an isometric folding $g:M \rightarrow {\mathbb R}^3$ ($f\neq g$) whose homotopy $H_t$ satisfies $\Sigma(H_t)=G$ for all $t \in (0,1]$ and every fold angle of the crease is strictly increasing or strictly decreasing?

\item
\emph{Rigid foldability with optional creases}:  Is there a nonempty subset $G'\subset G$ with a positive (``yes'') answer to the rigid foldability problem using all creases in~$G'$?
\end{enumerate}
\end{definition}

The rigid foldability using all creases question for single-vertex crease patterns is characterized in \cite{ACDEHKLT}.  
This paper considers the both former and latter questions, and we prove that it is weakly NP-hard for the former case and NP-hard for the latter case for a general crease pattern with multiple vertices.

Note that there are mechanical aspects and challenges to these questions that we are not considering here.  
It is known that in the flat-unfolded state of a rigidly-foldable crease pattern with at least one vertex will often lie at a singular point in the configuration space parameterized by the degrees of freedom of the model (e.g., see \cite{SilverbergScience}).  
The above foldability questions only consider if it is {\em possible} to travel in the configuration space of rigid folds from the unfolded state to some folded state, not whether it is mechanically feasible to, say, program a rigid, self-folding material to do so.  Thus, we are considering only the possibility of the so-called {\em finite rigid foldability from a planar state}.  

\section{Rigid Origami Basics}
\subsection{Quadrivalent Vertex}
We use \emph{quadrivalent flat-foldable vertex} as the building block of our gadgets.
A quadrivalent flat-foldable vertex is an origami vertex with four creases $c_i$ ($i=0,1,2,3 \mod 4$) in counterclockwise order with supplementary opposite sector angles, i.e.,
$$
(\theta_{0,1}, \theta_{1,2},\theta_{2,3},\theta_{3,0})=(\alpha, \beta, \pi-\alpha, \pi-\beta)
$$
($0<\alpha, \beta<\pi$) where $\theta_{i,i+1}$ is the sector angle between creases $c_i$ and $c_{i+1}$.
Let $\rho_{i}$ denote the fold angle of crease $c_i$, where fold angles are signed angles between the normals of the facets.

This type of origami vertex is known to satisfy the interesting and useful property that tangent of half fold angles $\tan{\rho_i \over 2} = t_i$ are proportional to each other \cite{Bricard,Huff,Hull,LangTwist15,Tachi:2010:bdffpq}.  Since this result is fundamental to our main result, we present an explicit proof not given in the references.

\begin{theorem}\label{thm:tan}
Any degree-4 flat-foldable vertex is rigidly foldable.
The configuration space of such a vertex, 
represented by the tangent of half the fold angles, i.e, $t_i=\tan(\frac{\rho_i}{2})$,
is the union of configurations satisfying: 
\begin{align}
\left(t_0,t_1,t_2,t_3\right) &=
\begin{cases}
\left(t, -p_a(\alpha,\beta) t, t, p_a(\alpha,\beta) t\right) & \textrm{Case (a)}\\
\left(-p_b(\alpha,\beta)t, t, p_b(\alpha,\beta) t, t \right) &  \textrm{Case (b)}
\end{cases}
\label{eq:folding-modes}
\end{align}
where $p_a$ and $p_b$ are constants defined by
\begin{align}
p_a\left(\alpha,\beta\right)
&= {1-\tan{\alpha\over 2}\tan{\beta\over 2}\over 1+\tan{\alpha\over 2}\tan{\beta\over 2}}
\label{eq:p-a}
\\
p_b\left(\alpha,\beta\right)
&= {\tan{\beta\over 2} -\tan{\alpha\over 2} \over \tan{\beta\over 2}+\tan{\alpha\over 2}}.
\label{eq:p-b}
\end{align}
\end{theorem}

Figure~\ref{fig:single-vertex} shows Cases (a) and (b).
Note that $0 \leq|p_a| <1$ and $0 \leq|p_b| <1$.
Also, if $\alpha$ is strictly smaller than other sector angles, $0< p_a  <1$ and $0<p_b<1$.

\begin{figure}[t]
	\includegraphics[width=\linewidth]{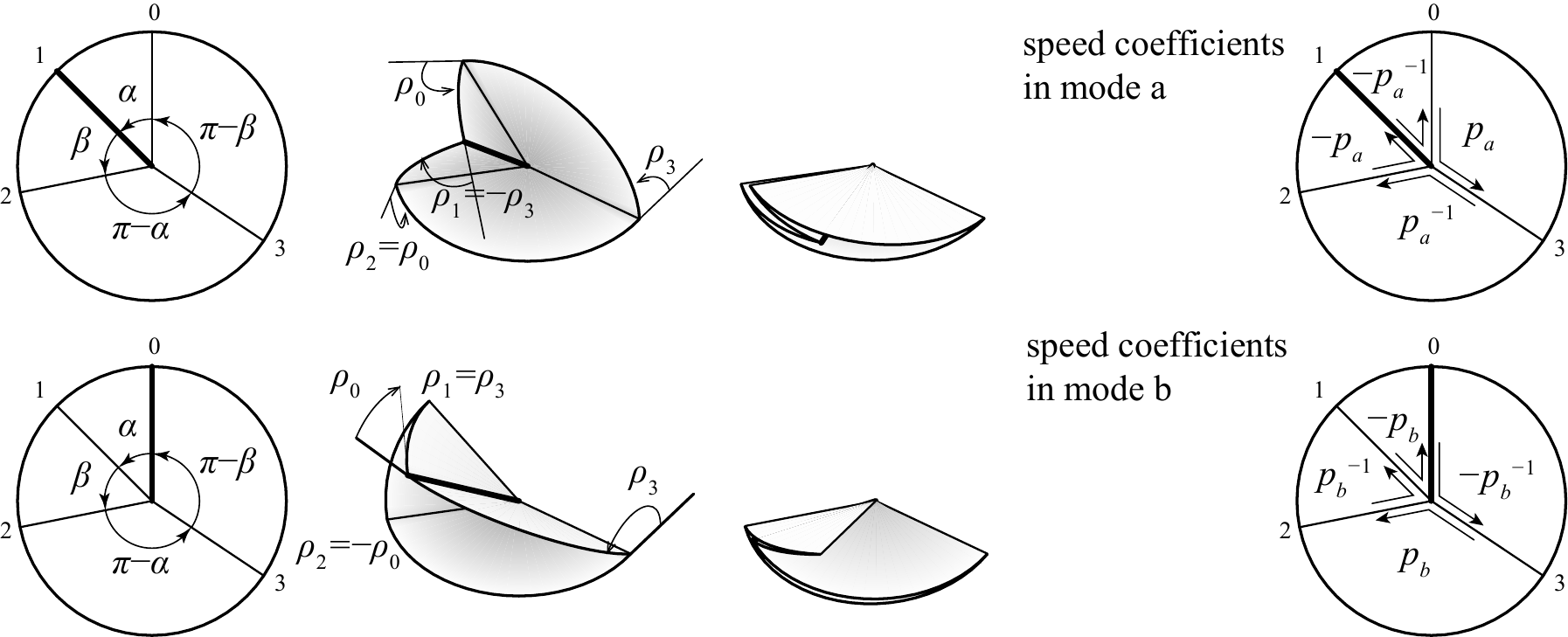}
	\centering
	\caption{Kinematics of degree-4 vertex with flat-foldability. Note that a single vertex has two modes of one-DOF motions (A) and (B).}
	\label{fig:single-vertex}
\end{figure}

To prove Theorem~\ref{thm:tan} we use the following Lemma:
\begin{lemma}\label{lemma1}
The fold angles of a quadrivalent flat-foldable vertex must satisfy
$$\rho_0=\rho_2\mbox{ and }\rho_1=-\rho_3\mbox{ or }\rho_1=\rho_3\mbox{ and }\rho_0=-\rho_2.$$
\end{lemma}

Lemma~\ref{lemma1} can be proven using the spherical law of cosines on the spherical polygon that the rigidly folded vertex cuts out of a sphere of radius 1 with the degree-4 vertex at its center.  (The spherical polygon needs to be triangulated by drawing a geodesic between the corners made by $\rho_1$ and $\rho_3$.  See \cite{Huff,Hull12} for details.)

\begin{proof} [Proof of Theorem~\ref{thm:tan}]

It is a straight-forward matter to verify that folding angles given in Equation~\eqref{eq:folding-modes} verify the rigid foldability condition in Equation~\eqref{eq1}, thus proving sufficiency of the  Theorem~\ref{thm:tan} conditions to rigidly fold a degree-4, flat-foldable vertex.  

To prove necessity, suppose we have a rigidly-folded state of the vertex.  Let $L_i$ denote the crease whose folding angle is $\rho_i$, and let us assume our vertex is at the origin, the unfolded vertex lies in the $xy$-plane in $\mathbb{R}^3$, and $L_3$ lies on the $x$-axis with $L_0$, $L_1$, and $L_2$ proceeding counterclockwise from there.  If we folded the vertex into our given rigid state while leaving the sector between $L_2$ and $L_3$ fixed in the $xy$-plane, then we may consecutively rotate each sector angle through, and each folding angle around, the positive $x$-axis to obtain a matrix equation as follows:  Let $R_x(\theta)$ and $R_z(\theta)$ denote the rotation matrices about the $x$- and $z$-axes by $\theta$, respectively. Then we have
$$R_z(\pi-\beta)R_x(\rho_0)R_z(\alpha)R_x(\rho_1)R_z(\beta)R_x(\rho_2)R_z(\pi-\alpha)R_x(\rho_3)=I.$$
We employ Lemma~\ref{lemma1} and bring some matrices to the other side to obtain
$$R_z(\pi-\beta)R_x(\rho_2)R_z(\alpha)R_x(-\rho_3)= R_x(-\rho_3)R_z(-(\pi-\alpha))R_x(-\rho_2)R_z(-\beta).$$
The 1st row, 2nd column entry of this matrix equation simplifies to
$$(\cos\rho_2-\cos\rho_3)\sin\alpha\cos\beta + (\cos\rho_2 \cos\rho_3 -1)\cos\alpha\sin\beta + \sin\rho_2 \sin\rho_3 \sin\beta=0.$$
We now let $t_2=\tan(\rho_2/2)$ and $t_3=\tan(\rho_3/2)$, which means that
$$\sin\rho_i = \frac{2t_i}{1+t_i^2}\mbox{ and }\cos\rho_i = \frac{1-t_i^2}{1+t_i^2}.$$
Substituting these into our condition and simplifying a bit we obtain
$$\frac{t_2^2-t_3^2}{(1+t_2^2)(1+t_3^2)}\sin\alpha\cos\beta + \frac{t_2^2+t_3^2}{(1+t_2^2)(1+t_3^2)}\cos\alpha\sin\beta
- \frac{2t_2t_3}{(1+t_2^2)(1+t_3^2)}\sin\beta=0. $$
Multiplying both sides by $(1+t_2^2)(1+t_3^2)$ and re-arranging further we get
$$2\sin\beta\frac{t_3}{t_2} = \left(1-\left(\frac{t_3}{t_2} \right)^2\right)\sin\alpha\cos\beta + \left( 1+\left(\frac{t_3}{t_2} \right)^2\right)\cos\alpha\sin\beta,$$
from which we solve for $t_3/t_2$ to obtain two solutions:
$$\frac{t_3}{t_2}  = \frac{\cos\left(\frac{\alpha+\beta}{2}\right)}{\cos\left(\frac{\alpha-\beta}{2}\right)}
= \frac{1-\tan\frac{\alpha}{2}\tan\frac{\beta}{2}}{1+\tan\frac{\alpha}{2}\tan\frac{\beta}{2}}\mbox{ and }
\frac{t_2}{t_3}=-\frac{\sin\left(\frac{\alpha+\beta}{2}\right)}{\sin\left(\frac{\alpha-\beta}{2}\right)}
=\frac{\tan\frac{\beta}{2}-\tan\frac{\alpha}{2}}{\tan\frac{\beta}{2}+\tan\frac{\alpha}{2}}.$$
Thus the folding angles of our rigidly folded state satisfy the Equation~\eqref{eq:folding-modes} formulas, where Lemma~\ref{lemma1} verifies that the mountain and valley creases must match one of the two modes shown in Figure~\ref{fig:single-vertex}.
\end{proof}

Theorem~\ref{thm:tan} implies that the fold angles of a degree-4, flat-foldable vertex are proportional to each other when parameterized by the tangent of the half angle.
Also, the configuration space curves determined by Equation~\eqref{eq:folding-modes} share exactly one configuration point $\vec t = \vec 0$ (the flat state).  At this point the configuration space branches to two rigid folding modes, each of which is one-DOF.
We can determine the folding motion of the vertex by specifying the mode and the fold angle of one of the creases.

\begin{definition}[Speed coefficient]
For a pair of adjacent creases $e_i$ and $e_j$, we call $p(e_i, e_j) := \tan\frac{\rho_{e_i}}{2}/ \tan\frac{\rho_{e_j}}{2}$ the \emph{speed coefficient} from $e_j$ to $e_i$.
\end{definition}
The speed coefficients are also known as the {\em fold-angle multipliers} \cite{LangTwist15}. By the choice of modes (a) and (b) in Equation~\eqref{eq:folding-modes}, the speed coefficient between creases is determined.
\begin{align*}
p(0, 1) &= -p_a^{-1} \textrm{ or } -p_b\\ 
p(1, 2) &= -p_a \textrm{ or } p_b^{-1} \\
p(2, 3) &= p_a^{-1} \textrm{ or } p_b\\
p(3, 0) &= p_a \textrm{ or } -p_b^{-1}
\end{align*}

Especially, for a special case of $\beta = 90\degrees$, the fold speed coefficient can be written using single variable
$$
p = p(\alpha) = p_a(\alpha,90\degrees) = p_b(\alpha,90\degrees) = {1 -\tan{\alpha\over 2} \over 1+\tan{\alpha\over 2}}.
$$
\begin{align*}
p(0, 1) &= -p^{-1} \textrm{ or } -p\\ 
p(1, 2) &= -p \textrm{ or } p^{-1} \\
p(2, 3) &= p^{-1} \textrm{ or }  p\\
p(3, 0) &= p \textrm{ or } -p^{-1}
\end{align*}
Here, if we chose $\alpha = \arctan {3\over 4}$, then we obtain the folding speed coefficient $p = {1 \over 2}$.

Theorem~\ref{thm:tan} gives us a continuously parameterized family of rigid origami foldings for degree-$4$, flat-foldable vertices.  Thus we have the following.
\begin{lemma}\label{lemma-assign}
There are exactly $4$ valid MV assignments (including zero assignment for the crease not folded) for each vertex from the flat state, and each assignment uniquely restricts the configuration space to a $1$ parameter curve, i.e., $1$-DOF mechanism.
\end{lemma}
\begin{proof}
The two binary options, (1) choice of modes (a) and (b) and (2) the sign of $t$, give four folding paths, along which the signs of fold angles (assignment) are unchanged by Equation~\eqref{eq:folding-modes}.
The opposite paths in the same mode have a reversed MV assignment of the other, and the paths in different modes have different assignment because edges $0$ and $3$ have the same (non-zero) assignment in mode (a) and the opposite assignment in mode (b).
So, each MV assignment gives a unique $1$-parameter folding path in the configuration space.
\end{proof}

\subsection{Assembly: Flat-foldable Quadrivalent Origami}

We call a pattern composed solely of flat-foldable degree-$4$ vertices a \emph{flat-foldable quadrivalent mesh}.
For our hardness proof of rigid origami, we will use flat-foldable quadrivalent mesh (with optional creases).
Determining the rigid foldability of flat-foldable quadrivalent meshes is a problem of assigning modes (a) or (b) for each vertex in such a way that  the folding speeds of the creases will not be in conflict.  Note that because we are only interested in rigidly folding our crease pattern a finite amount from the flat, unfolded state, we do not need to worry about the possibility of different parts of the paper colliding or self-intersecting.

The following directly follows from Lemma~\ref{lemma-assign} and the definition of  speed coefficients.  (See also \cite{LangTwist15}.)

\begin{corollary}[Assignment Problem]
For a MV assignment of flat-foldable quadrivalent mesh, there is at most one folding path that forms a $1$-manifold in the configuration space. 
\end{corollary}

\begin{corollary}[Closure Condition] \label{th:degree4closure}
A MV assignment of flat-foldable quadrivalent mesh yields a rigid folding motion if and only if for each $k$-gonal face $f$ surrounded by creases $c_i$ ($i=0, \dots, k-1$), the speed coefficients $p_i = p(c_{i+1}, c_i)$ satisfy the following:
\begin{equation}
\prod_{i=0, \dots, k-1} p_i = 1.
\label{eq:degree4closure}
\end{equation}
\end{corollary}

Equation~\ref{eq:degree4closure} provides a natural metric for measuring
error in a candidate rigid folding motion.

\begin{definition}[Finite-Precision Degree-4 Flat-Foldable Rigid Foldability]
\label{approximate definition}
As in Definition~\ref{exact definition}, we define two decision problems
whose primary input is a straight-line planar graph drawing $G$ on a polygon
$M\subset{\mathbb R}^2$,
where all vertex coordinates are specified by rational numbers.
Now $G$ is constrained to be flat foldable and have all vertices of degree~$4$,
and we have an additional input $\epsilon > 0$ (specified as a rational number).

\begin{enumerate}
\item
\emph{Rigid foldability using all creases}:  Is there a mountain--valley assignment whose induced fold-angle multipliers satisfy, at each vertex,
\begin{equation}
\label{eq:degree4closure-epsilon}
\prod_{i=0, \dots, k-1} p_i \in [1-\epsilon,1+\epsilon]?
\end{equation}

\item \emph{Rigid foldability with optional creases}:  Is there a nonempty subset $G'\subset G$ with a positive (``yes'') answer to the finite-precision rigid foldability problem using all creases in~$G'$?
\end{enumerate}
\end{definition}

\begin{theorem}[Finite-Precision Rigid Foldability is in NP]
\label{deg4-np}
Both finite-precision degree-4 flat-foldable rigid foldability problems
are in NP, even with $\epsilon = 0$ (exact precision).
\end{theorem}
\begin{proof}
First observe that we can confirm that the given crease pattern is flat
foldable in~P.  At each vertex $\vec u$ with neighbors
$\vec v_1,\vec v_2,\vec v_3,\vec v_4$, we reflect $\vec v_1 - \vec u$ through
$\vec v_i - \vec u$ for $i \in \{2,3,4\}$, and then verify that
the resulting vector is identical to $v_1 - u$.
Here we use that the reflection $\vec p'$ of vector $\vec p$ through another
vector $\vec q$ can be done with $O(1)$ additions, multiplications, and
divisions:
\begin{equation} \label{eq:reflection}
  \vec p' = \vec p - 2 {(\vec p \cdot \vec q^\perp) \vec q^\perp \over
                        \vec q \cdot \vec q},
  \quad \text{where} \quad
  (x,y)^\perp = (-y,x).
\end{equation}

Next we nondeterministically guess the mountain--valley assignment and,
for the optional-crease problem, the subset $G'$ of creases.

It remains to verify Inequality~\eqref{eq:degree4closure-epsilon}
at each vertex.  Equations~\eqref{eq:p-a} and~\eqref{eq:p-b}
represent each $p_i$ term in the product as a rational function of tangents
of half-angles.  Recall the tangent half-angle formula:
$$
\tan {\theta \over 2} = \csc \theta - \cot \theta .
$$
The latter trigonometric functions can be represented by radical expressions
on the vertex coordinates, specifically, their absolute difference in $x$
coordinates, their absolute difference in $y$ coordinates, and their
Euclidean distance (which involves a square root).
Therefore the product of the $p_i$'s at a vertex can be represented as a
constant-complexity expression involving addition, subtraction, multiplication,
division, and square roots on the input rationals.
Comparing such an expression to $1 \pm \epsilon$ (even with $\epsilon = 0$)
can be done exactly in polynomial time; see, e.g., \cite{Burnikel2009}.
\end{proof}
Note that rigid foldability of general origami with non-degree-$4$ or non-flat-fordable vertices may not be in NP.

\subsection{The Square Twist Fold}
By combining four degree-$4$ vertex, we may obtain a rigid folding version of the ``square twist fold.'' 
Unlike the original square twist fold \cite{fujimoto,Hull12}, we choose the mountain valley assignment to allow rigid foldability.
Such a system can fold in two ways if we take the symmetry into account (\cite{LangTwist15,Tachi-Hull:2016}).

\begin{figure}[t]
	\includegraphics[width=.8\linewidth]{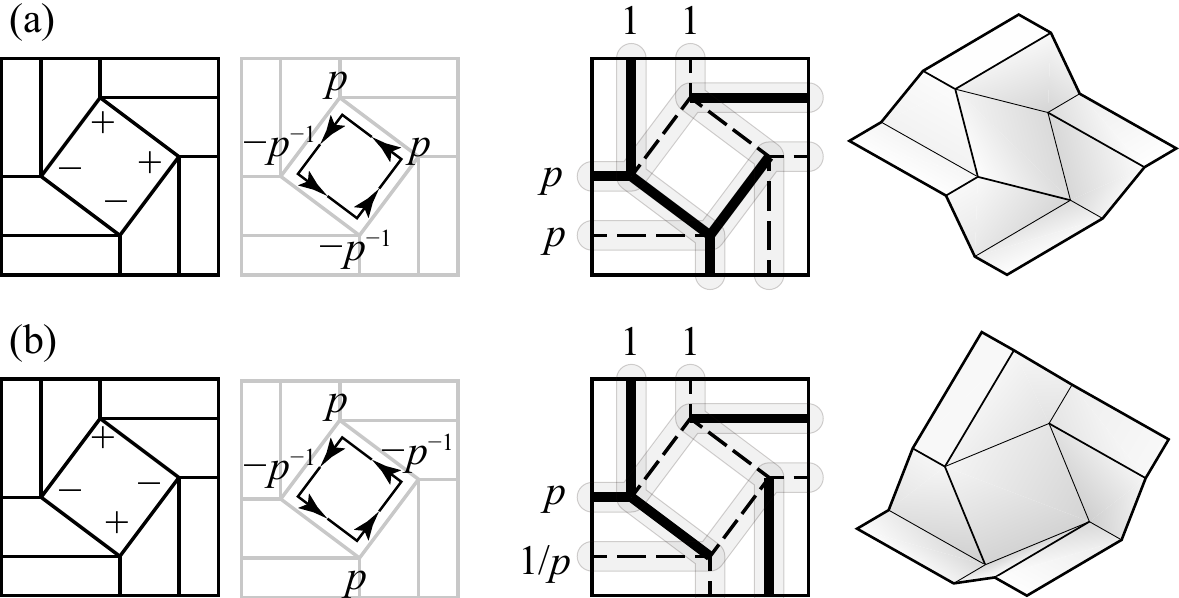}
	\centering
	\caption{Two different modes (a) and (b) of a square twist. Left two columns show the assignment of modes $+$ and $-$ to obtain the consistent loop, which gives the MV assignments in third column, which is folded to 3D state in the fourth column.
	The number $1$, $p$, $1/p$ represents the absolute folding speeds measured in tangent of half of fold angles of edges along the chain of opposite edges.}
	\label{fig:twist}
\end{figure}

\begin{lemma}\label{twist-mode}
There are only $4$ MV assignments and thus $2$ modes (up to symmetry) that allow the square twist to fold rigidly.  
They are the ones shown in Figure~\ref{fig:twist} along with the relative folding speed coefficients of the creases on the boundary.
\end{lemma}

\begin{proof}
Consider the central square composed of creases $c_0,c_1,c_2,c_3$.
The problem of assigning MV is equivalent to assigning the coefficient 
$$
p_{c_i,c_{i+1}} = 
\begin{cases}
p(\alpha) & \textrm{mode $+$}\\
-\frac 1 p(\alpha) & \textrm{mode $-$}
\end{cases},
$$
where $\alpha$ is the twist angle of the center square, i.e., the smallest sector angle at each vertex.
Equation~\eqref{eq:degree4closure} is satisfied if and only if the number of mode $+$'s is $2$.
The possible pattern is either $++--$ or $+-+-$ up to symmetry, which is as described in Figure~\ref{fig:twist}.
\end{proof}

\section{Rigid Foldability using All Creases is Weakly NP-hard}

We show how to reduce the Partition problem to finite-precision rigid
foldability using all creases, thereby establishing weak NP-hardness.

\begin{theorem}[Rigid Foldability is weakly NP-hard]
\label{usingall}
Finite-precision degree-4 flat-foldable rigid foldability using all the creases
is weakly NP-hard.
\end{theorem}
\begin{figure}[tbhp]
	\includegraphics[width=\linewidth]{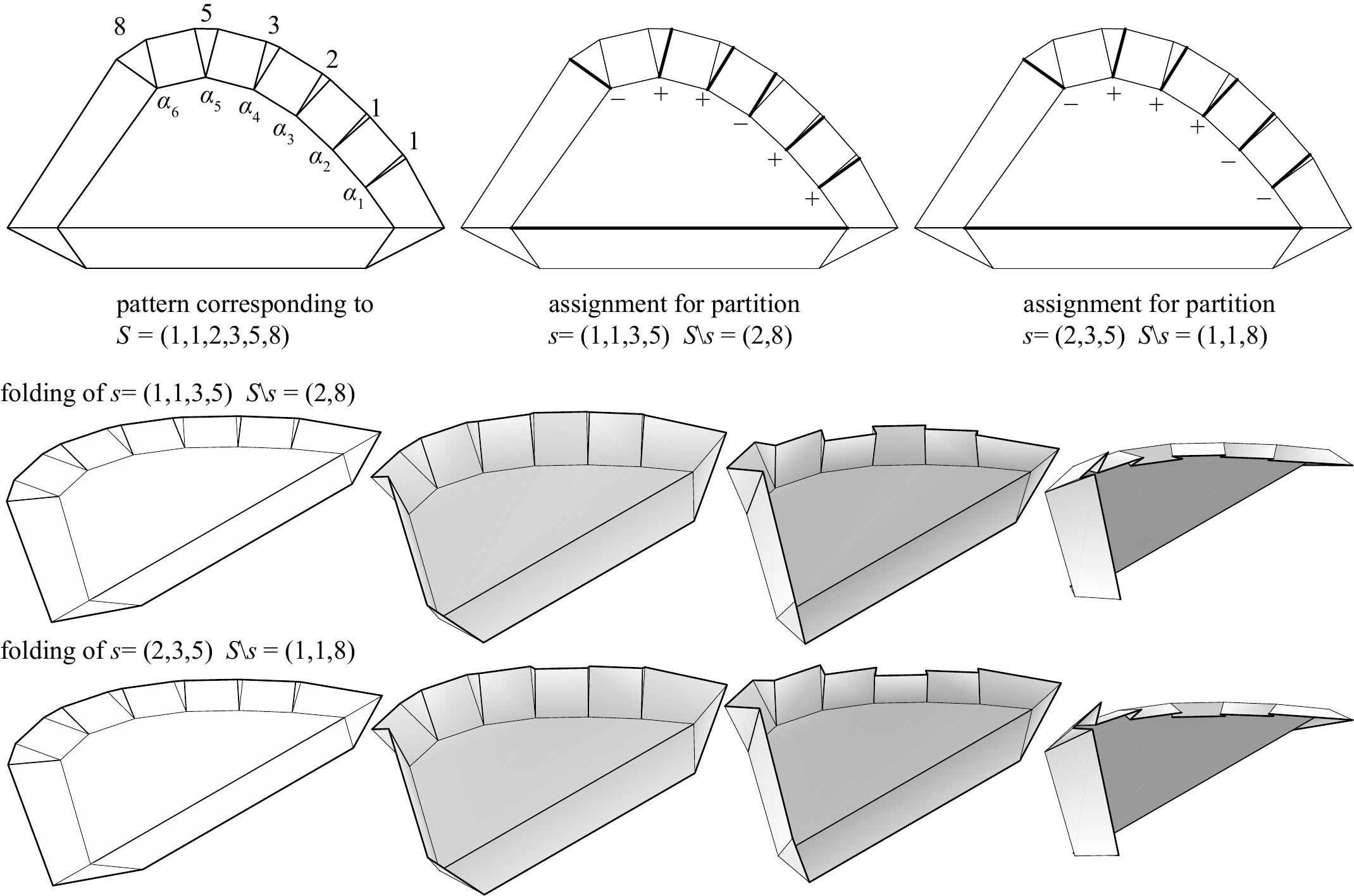}
	\centering
	\caption{Rigid origami realizing a given partitioning problem of $n>1$ by $n+2$-gonal closed chain of vertices.
	 The sector angles $\alpha_i$ of consecutive $n$ vertices are designed from the set of integers (two left figures),
	 and two vertices form a mirror symmetric degree-$4$ vertices.
	 Partitioning corresponds to MV assignments for each vertex (right two figures). 
	 The pattern is rigidly foldable when the MV pattern partitions the set into equal sums (bottom).}
	\label{fig:using-all-creases}
\end{figure}

\begin{proof}
To demonstrate the main idea, we first pretend that we can compute the needed real values exactly, and second we take care of finite precision.

Consider a partition problem of finding a subset $S$ of $A = \{a_1, a_2, \dots, a_n\}$ ($a_i \in \mathbb N$, $n> 4$) that partition the total values into halves, i.e., $\Sigma(S) = \sum(A \setminus S) = {1 \over 2} \Sigma(A)$, where $\Sigma(S) = \sum_{x\in S} x$.
We construct a crease pattern of an $(n+2)$-gonal closed chain of
flat-foldable, quadrivalent vertices $\mathbf v_0, \mathbf v_1, \dots, \mathbf v_{n+1}$;
refer to Figure~\ref{fig:using-all-creases}.
We start with $n$ consecutive flat-foldable vertices
$\mathbf v_1, \mathbf v_2, \dots, \mathbf v_n$
where $\mathbf v_i$ has
$\alpha = \alpha_i\in(0^\circ,90^\circ)$ and $\beta=90\degrees$,
so that $\sum_i\alpha_i < 360^\circ$.
We split the remaining angle $360^\circ - \sum_i\alpha_i$ into halves, and make the last two vertices $\mathbf v_0$ and $\mathbf v_{n+1}$ be mirror symmetric and flat-foldable with $\alpha = \beta= \alpha_r = (360^\circ - \sum_i\alpha_i)/2$.

Now we consider the speed coefficients in a rigid folding.
For $1 \leq i \leq n$, vertex $\mathbf v_i$ has a speed coefficient of $p_i = p(\alpha_i) = \frac{1-\tan\frac{\alpha_i}{2}}{1+\tan\frac{\alpha_i}{2}}$ (mode $+$) or $p_i^{-1}$ (mode $-$).
Vertex $\mathbf v_{n+1}$ has speed coefficient of $p_r=p(\alpha_r,\alpha_r)=\frac{1-\tan^2\frac{\alpha_r}{2}}{1+\tan^2\frac{\alpha_r}{2}}$ (mode $+$) and $\infty$ (mode $-$), and vertex $\mathbf v_0$ has speed coefficient of $0$ (mode $+$) and $p_r^{-1}$ (mode $-$).
In order that all creases of vertices $\mathbf v_{n+1}$ and $\mathbf v_0$ fold simultaneously, $\mathbf v_{n+1}$ must chose mode $+$ and $\mathbf v_0$ must chose mode $-$, because $p_r$ is nonzero and finite.

Now we tweak $\alpha_i$ properly so that choosing the vertices that folds in mode $-$ is equivalent to choosing the element of $S$ from $A$, and the closure constraint (Equation~\eqref{eq:degree4closure}) is equivalent to the partition problem.
Specifically, we set the variables $\alpha_i$
such that
\begin{align}
{a_i \over \Sigma(A)} &= \log p_i^{-1},
\end{align}
for all $i=1,\ldots, n$, and thus
\begin{align}\label{eq:angles}
\tan \frac{\alpha_i}{2} &=   -1 + \frac{2}{1+e^{-\frac{a_i}{\Sigma(A)}}} =  \tanh\left(\frac{a_i}{2 \Sigma(A)}\right).
\end{align}
Then, by Equation~\eqref{eq:degree4closure}, the formed pattern is rigidly foldable if and only if
there exists a subset $S\subset A$ such that
\begin{align}
\left(\prod_{i\in S} p_i^{-1} \right)\left(\prod_{i\in A\setminus S} p_i\right) p_r p_r^{-1}&= 1,
\intertext{which, by taking the logarithm on both sides and multiplying by $\Sigma(A)$, is equivalent to the given partition problem:
}
\Sigma(S) - \Sigma(A \setminus S) = 0.
\end{align}
Because we scaled by $1/\Sigma(A)$, we have
\begin{align}
\sum_i \alpha_i &= \sum_i 2\arctan \left(\tanh\left(\frac{a_i}{2 \Sigma(A)}\right)\right)
< 2\sum_i \frac{a_i}{2 \Sigma(A)} = 1 < \frac{\pi}{3} < \pi.
\end{align}
Thus, $\sum_{i\in \{1,2,\dots,n\}} \alpha_i \in (0^\circ, 180^\circ)$,
so $\alpha_r(\Sigma(A)) =(360^\circ - \sum_{i\in \{1,2,\dots,n\}} \alpha_i)/2 \in (90^\circ, 180^\circ)$.

Notice that this ensures the range $p_i \in (0.26, 1)$, useful later for avoiding close-to-singular vertices. 
This also ensures that each turn angle of the central polygon is within $(0, 60^\circ)\subset(0, 180^\circ)$ and thus the convexity of the central polygon.
Thus we can always construct the $(n+2)$-gon from the slope of each segment derived from these turn angles.
Also, the range $\alpha_r \in(120^\circ, 180^\circ)$ ensures that $p_r \in (-1,-0.26)$, in particular, non-zero and finite, proving the assumption.
 
Next we handle the precision issues.
The reduction algorithm first computes $\tilde t_i$,
an approximation to $t_i = \tan {\alpha_i \over 2}$ given
by Equation \eqref{eq:angles}, but with precision $\epsilon_t$.
We can compute $x = \frac{a_i}{2 \Sigma(A)}$ exactly, represented as a rational.
To then compute $\tanh(x)$, we can use that it is the solution to the
autonomous differential equation $y' = 1 - y^2$ where $y(0) = 0$
\cite{tanh}, which we can approximate using the Runge--Kutte method
or even Euler's method.

\begin{figure}[tbhp]
	\includegraphics[width=0.5\linewidth]{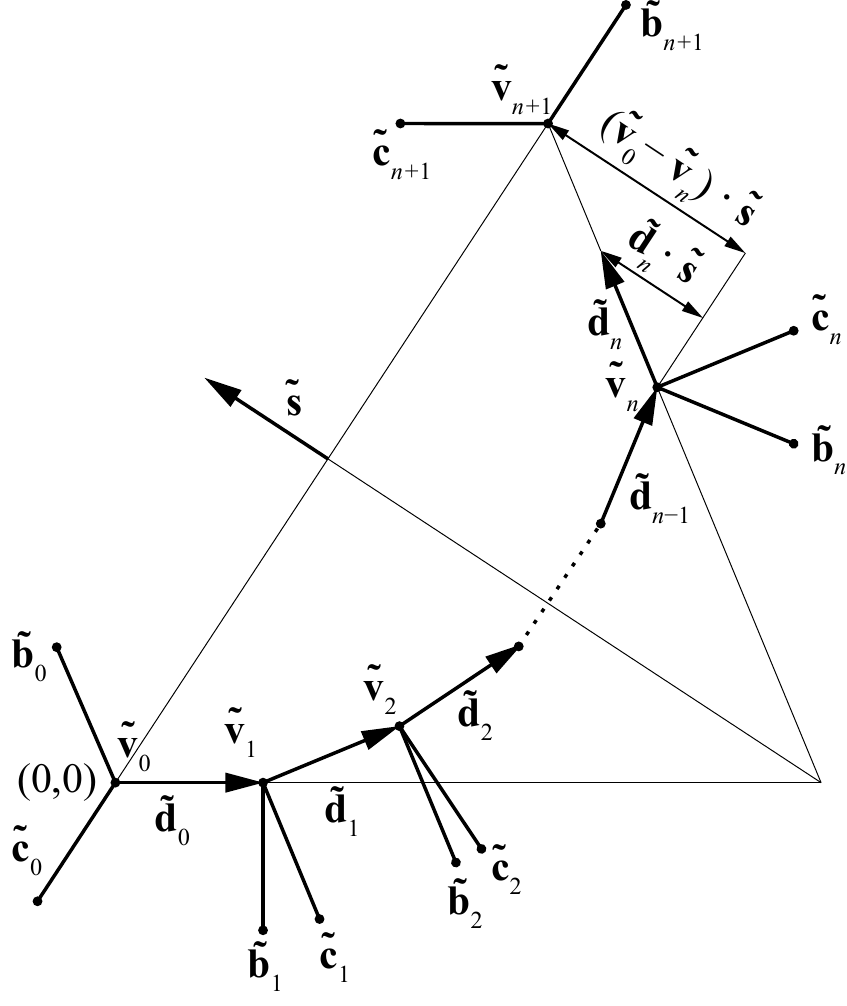}
	\centering
	\caption{The actual construction step of rigid origami realizing partition problem.}
	\label{fig:using-all-construction}
\end{figure}

Next the algorithm computes an approximation of the crease pattern,
one vertex at a time,
with computed vertex $\tilde{\mathbf v}_i$ approximating ideal vertex
$\mathbf v_i$.
Refer to Figure~\ref{fig:using-all-construction}.
Vertices $\tilde{\mathbf v}_0, \tilde{\mathbf v}_1, \dots, \tilde{\mathbf v}_n$
will lie on a grid with resolution $\epsilon_g = 1/G$
where $G$ is an integer to be chosen later.
First we compute \emph{difference vectors} $\tilde{\mathbf d}_i$
for $0 \leq i \leq n$, whose coordinates are also multiples of~$\epsilon_g$.
For $i < n$, $\tilde{\mathbf d}_i$ corresponds to
$\tilde{\mathbf v}_{i+1} - \tilde{\mathbf v}_i$
(but $i=n$ will behave slightly differently),
so we can compute $\tilde{\mathbf v}_i$ for $1 \leq i \leq n$ via
$\tilde{\mathbf v}_{i+1} = \tilde{\mathbf v}_i + \tilde{\mathbf d}_i$,
using $\tilde{\mathbf v}_0 = (0,0)$ as a base case.
We define $\tilde{\mathbf d}_0 = (1,0)$,
the unit rightward vector.
Given $\tilde{\mathbf d}_{i-1}$, we compute $\tilde{\mathbf d}_i$ by rotating
the vector $\tilde{\mathbf d}_{i-1}$ by the rotation matrix
\begin{equation}\label{eq:rotmat}
 \begin{pmatrix}
     1-\tilde t_i^2 \over 1+\tilde t_i^2 & -2\tilde t_i \over 1+\tilde t_i^2
     \\[2ex]
     2\tilde t_i \over 1+\tilde t_i^2    & 1-\tilde t_i^2 \over 1+\tilde t_i^2
   \end{pmatrix},
\end{equation}
and rounding the resulting coordinates to the nearest integer
multiples of~$\epsilon_g$.
Finally, we compute $\tilde{\mathbf v}_{n+1}$.
We want $\tilde{\mathbf v}_{n+1}$ and $\tilde{\mathbf v}_0$ to be
exactly mirror symmetric through the bisecting vector $\tilde{\mathbf s}$ of
$\tilde{\mathbf d}_n$ and $-\tilde{\mathbf d}_0$.
This bisecting vector is given by
$\tilde{\mathbf s} = {1 \over 2} (\tilde{\mathbf d}_n - \tilde{\mathbf d}_0)$.
We check that $\tilde{\mathbf v}_n + \tilde{\mathbf d}_n$
has a smaller dot product than $-\tilde{\mathbf v}_0$ does with
the bisecting vector, so that we can make
$\tilde{\mathbf v}_{n+1} - \tilde{\mathbf v}_n$ longer than $1$ unit.
If this condition does not hold, we restart the entire algorithm with the
order $a_1, a_2, \dots, a_n$ reversed.
Then we set $\tilde{\mathbf v}_{n+1}$ to
$$
  \tilde{\mathbf v}_n +
  {
    (\tilde{\mathbf v}_0 - \tilde{\mathbf v}_n) \cdot \tilde{\mathbf s}
    \over
    \tilde{\mathbf d}_n \cdot \tilde{\mathbf s}
  } ~ \tilde{\mathbf d}_n
$$
(\emph{not} rounded to the $\epsilon_g$-resolution grid).

Now we compute the boundary vertices
$\tilde{\mathbf b}_i, \tilde{\mathbf c}_i$
of the piece of paper, connected via creases to $\tilde{\mathbf v}_i$,
for each $i = 0, 1, \dots, n+1$.
For $i = 1, 2, \dots, n$, we use equations
$\tilde{\mathbf b}_i - \tilde{\mathbf v}_i =
 -\tilde{\mathbf d}_{i-1}^\perp$ and
$\tilde{\mathbf c}_i - \tilde{\mathbf v}_i =
 -\tilde{\mathbf d}_i^\perp$
(i.e., $90^\circ$ right rotation of vectors $\tilde{\mathbf d}_{i-1}$
and $\tilde{\mathbf d}_i$),
which even remains on the $\epsilon_g$-resolution grid.
For $i=0$ and $i=n+1$, we use an extension of the segment
$\tilde{\mathbf v}_0 \tilde{\mathbf v}_{n+1}$ to define one boundary vertex,
and use reflections to compute the other boundary vertex.
Precisely, we define $\tilde{\mathbf c}_0$ and $\tilde{\mathbf b}_{n+1}$ via
$$
\tilde{\mathbf c}_0 - \tilde{\mathbf v_0} =
\frac{\tilde{\mathbf v}_0 - \tilde{\mathbf v}_{n+1}}
     {\lceil \|\tilde{\mathbf v}_0 - \tilde{\mathbf v}_{n+1}\|\rceil},
\quad
\tilde{\mathbf b}_{n+1} - \tilde{\mathbf v}_{n+1} =
\frac{\tilde{\mathbf v}_{n+1} - \tilde{\mathbf v}_0}
     {\lceil \|\tilde{\mathbf v}_{n+1} - \tilde{\mathbf v}_0\|\rceil};
$$
and define $\tilde{\mathbf b}_0$ and $\tilde{\mathbf c}_{n+1}$ via
$\tilde{\mathbf b}_0 - \tilde{\mathbf v}_0$ equalling the reflection of
$\tilde{\mathbf d}_0$ through $\tilde{\mathbf v}_0 \tilde{\mathbf v}_{n+1}$,
and
$\tilde{\mathbf v}_{n+1} - \tilde{\mathbf c}_{n+1}$ equalling the reflection of
$\tilde{\mathbf d}_n$ through $\tilde{\mathbf v}_0 \tilde{\mathbf v}_{n+1}$.
These coordinates are rational, with numerators and denominators bounded by
a polynomial in $1/\epsilon_g$.

Now we prove a sequence of claims about the constructed crease pattern
in order to show it is rigidly foldable up to precision~$\epsilon$.

First we claim that the crease pattern is exactly locally flat foldable,
i.e., satisfies Kawasaki's Theorem.  This claim follows by construction:
for each vertex $\tilde{\mathbf v}_i$, we constructed
$\tilde{\mathbf b}_i$ and $\tilde{\mathbf c}_i$ exactly to
guarantee local flat foldability of $\tilde{\mathbf v}_i$.
For $i = 1, 2, \dots, n$, we used $90^\circ$ rotations to guarantee two
right angles; and for $i = 0$ and $i=n+1$, we used exact reflections.

Next we claim that vertices $\tilde{\mathbf v}_0$ and $\tilde{\mathbf v}_{n+1}$
are mirror symmetric with each other, and each have two collinear creases
($\alpha = \beta$), as needed by the partition reduction.
This claim follows from the construction: the exact computation of the
bisecting vector $\tilde{\mathbf s}$ and vertices
$\tilde{\mathbf c}_0$ and $\tilde{\mathbf b}_{n+1}$.

Next we claim that each vertex $\tilde{\mathbf v}_i$, for $1 \leq i \leq n$,
has $\beta=90^\circ$ (exactly) and $\alpha = \tilde \alpha_i$ satisfying
$\tan \frac{\tilde \alpha_i}{2} = \tilde t_i \pm O(\epsilon_g)
= t_i \pm O(\epsilon_g + \epsilon_t)$.
That $\beta=90^\circ$ follows from the construction of
$\tilde{\mathbf b}_i$ and $\tilde{\mathbf c}_i$.
To measure $\tilde \alpha_i$, we first analyze the lengths of $\tilde{\mathbf d}_i$:
$\|\tilde{\mathbf d}_0\| = 1$,
and $\tilde{\mathbf d}_i$ is an $\epsilon_g$-grid rounding of a rotation of
$\tilde{\mathbf d}_{i+1}$, so
$| \, \|\tilde{\mathbf d}_i\| - \|\tilde{\mathbf d}_{i+1}\| \, |
\leq \sqrt 2 \, \epsilon_g$,
and therefore
$\|\tilde{\mathbf d}_i\|$ is within $\pm n \sqrt 2 \epsilon_g$ of~$1$.
Assuming $\epsilon_g < 1/(2 n \sqrt 2)$,
$\|\tilde{\mathbf d}_i\| \geq \frac{1}{2}$.
By construction, the angle $\tilde \alpha_i$ at $\tilde{\mathbf v}_i$ is equal to
the angle between vectors $\tilde{\mathbf d}_i$ and $\tilde{\mathbf d}_{i-1}$.
Because these vectors have length $\geq \frac{1}{2}$, the angle change
caused
by rounding the rotated $\tilde{\mathbf d}_{i-1}$ is at most
$2 \sqrt 2 \epsilon_g$.
By Taylor Series,
$$
\left|\tan {\tilde \alpha_i \over 2} - \tilde t_i\right|
\leq \left(2 \sqrt 2 \epsilon_g\right) \sec^2 \arctan \tilde t_i
   + \left(2 \sqrt 2 \epsilon_g\right)^2 \tilde t_i \sec^2 \arctan \tilde t_i
   + \cdots.
$$
Because $\alpha_i \in (0, 90^\circ)$, we have $t_i \in (0,1)$ and
thus $\tilde t_i \in [0,1]$, and $\sec^2 \arctan \tilde t_i \in [1,2]$.
Therefore $\left|\tan {\tilde \alpha_i \over 2} - \tilde t_i\right| = O(\epsilon_g)$,
and $\left|\tan {\tilde \alpha_i \over 2} - t_i\right| = O(\epsilon_g + \epsilon_t)$,

Finally we claim that
$\prod \tilde p_i = \prod p_i \pm O(n \epsilon_g + n \epsilon_t)$, where
$\tilde p_i = {1 - \tan^2 {\tilde \alpha_i \over 2} \over
               1 + \tan^2 {\tilde \alpha_i \over 2}}$.
We have $\tilde p_i = p_i \pm O(\epsilon_g + \epsilon_t)$
by a similar argument to the above.
Because $p_i \in (0.26, 1)$, we can convert this absolute error bound into
a relative error bound: $\tilde p_i / p_i = 1 \pm O(\epsilon_g + \epsilon_t)$.
Therefore
\begin{equation} \label{tilde p_i}
\prod \tilde p_i = \prod p_i \pm O(n (\epsilon_g + \epsilon_t)).
\end{equation}

This error bound gives us a lower bound on the precision $\epsilon$
in our output instance of finite-precision degree-$4$ flat-foldable
rigid foldability.
Namely, if
$\epsilon > \epsilon_{\mathrm{LB}} := c n (\epsilon_g + \epsilon_t)$,
where $c$ is the constant in the $O$ notation in Equation~\eqref{tilde p_i},
then the finite-precision rigid foldability instance has a ``yes'' answer
whenever $\prod \tilde p_i = 1$, i.e., when the Partition instance has
a ``yes'' answer.

It remains to prove that, if the Partition instance has a ``no'' answer,
then the constructed crease pattern is a ``no'' instance to finite-precision
rigid foldability.
For this property to hold, we need an upper bound on~$\epsilon$.
Consider a candidate partition $(S, A \setminus S)$, and suppose that
$\Sigma(S) \neq \Sigma(A\setminus S)$.  Because $A \subset \mathbb N$,
$|\Sigma(S) - \Sigma(A\setminus S)| \geq 1$.
Thus, for Inequality~\ref{eq:degree4closure-epsilon} to (incorrectly) hold
on the constructed crease pattern, we would need $\epsilon$ to be at least
\begin{align*}
\epsilon'_{\mathrm{UB}}
&= \left| \prod_i \tilde p_i -1 \right| \\
&= \left| \prod_i p_i -1 \pm O(n(\epsilon_g + \epsilon_t)) \right| \\
&= \left| \frac{\exp(\Sigma(S)/\Sigma(A))}{\exp(\Sigma(A\setminus S)/\Sigma(A))} -1 \pm O(n(\epsilon_g + \epsilon_t)) \right| \\
&= \left| \exp\left(\frac{\Sigma(S)-\Sigma(A\setminus S)}{\Sigma(A)}\right)-1  \pm O(n(\epsilon_g + \epsilon_t)) \right| \\
&\geq \left|\exp \left( {1 \over \Sigma(A)} \right) - 1 \pm O(n(\epsilon_g + \epsilon_t))\right| \\
&\geq \left|\frac{1}{\Sigma(A)} \pm O(n(\epsilon_g + \epsilon_t))\right|,
\end{align*}
by Taylor expansion.
Assuming $c n(\epsilon_g + \epsilon_t) < \frac{1}{\Sigma(A)}$ where $c$
is the (same) constant in the $O$ notation, the absolute value operation is
unnecessary, and we obtain
$$\epsilon'_{\mathrm{UB}} \geq
  \frac{1}{\Sigma(A)} - c n(\epsilon_g + \epsilon_t)
  =: \epsilon_{\mathrm{UB}}.$$
By setting $\epsilon < \epsilon_{\mathrm{UB}}$ (and thus
$< \epsilon'_{\mathrm{UB}}$),
we guarantee that ``no'' answers are preserved.

To guarantee that $\epsilon_{\mathrm{LB}} < \epsilon_{\mathrm{UB}}$,
we need that
$c n (\epsilon_g + \epsilon_t) <
\frac{1}{\Sigma(A)} - c n(\epsilon_g + \epsilon_t)$, i.e.,
$2 c n (\epsilon_g + \epsilon_t) < \frac{1}{\Sigma(A)}$.
Thus it suffices to set
$$\epsilon_g = \epsilon_t = \frac{1}{5 c n \Sigma(A)},
  \quad \text{i.e.,} \quad G = 5 c n \Sigma(A),$$
which we can compute and represent exactly as a rational number by taking an
integer upper bound on~$c$.

Therefore, by setting
$\epsilon = {\epsilon_{\mathrm{LB}} + \epsilon_{\mathrm{UB}} \over 2}$
(also computable exactly as a rational number),
we obtain $\epsilon_{\mathrm{LB}} < \epsilon < \epsilon_{\mathrm{UB}}$
as needed for correctness of the reduction.
All (rational) coordinates in the reduction have numerator and denominator
bounded by a constant-degree polynomial in $1/\epsilon_g = 5 c n \Sigma(A)$,
and thus so does the output number~$\epsilon$.
\end{proof}

\section{Rigid Foldability with Optional Creases is Strongly NP-hard}
\begin{theorem}
\label{optional}
Rigid foldability of a given crease pattern (with optional creases) is strongly NP-hard.
\end{theorem}
\begin{figure}[tbhp]
	\includegraphics[width=0.8\linewidth]{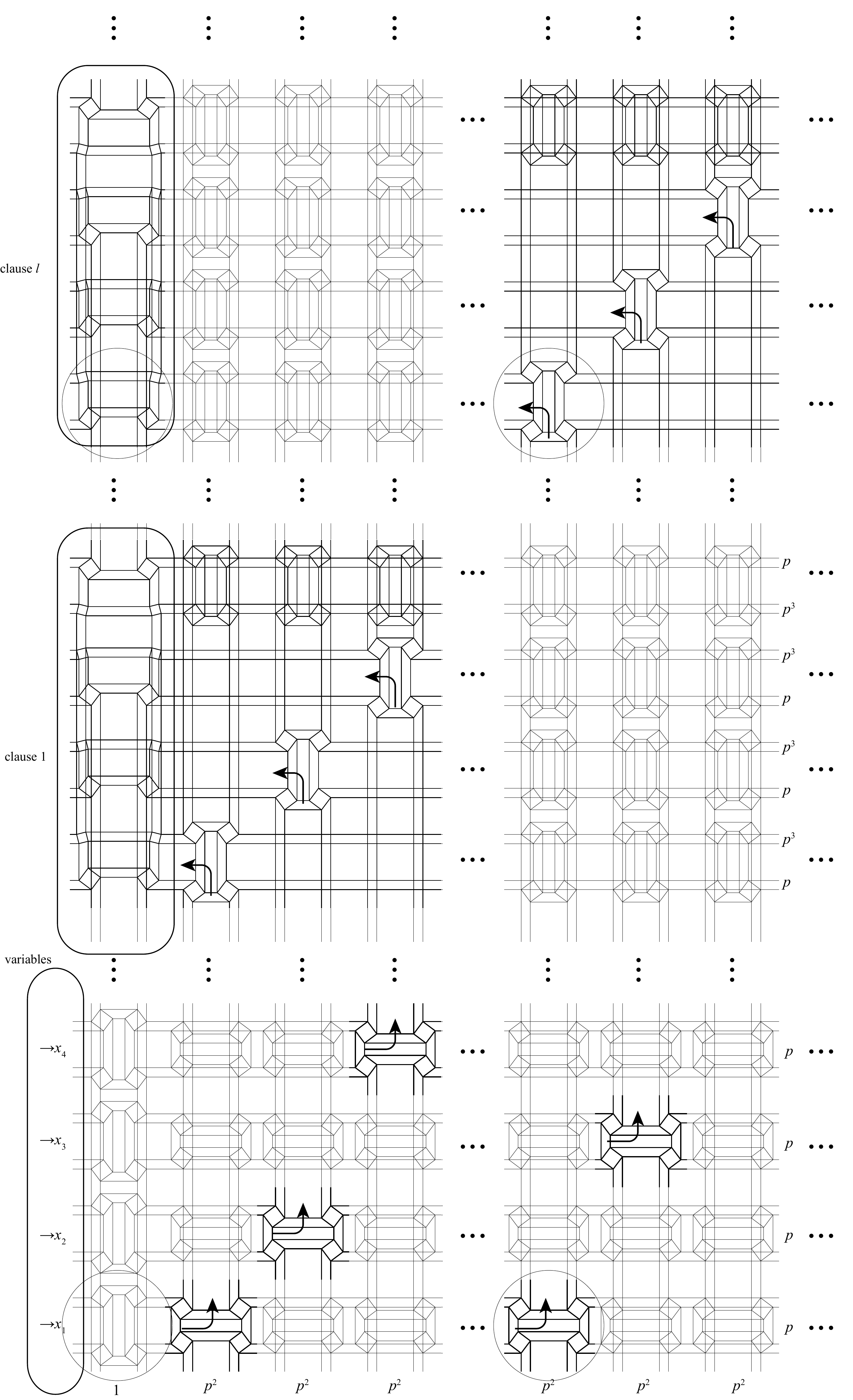}
	\centering
	\caption{A crease pattern rigidly foldable if and only if it satisfies a 1-in-3 SAT.}
	\label{fig:reduction}
\end{figure}

\begin{figure}[tbhp]
	\includegraphics[width=0.8\linewidth]{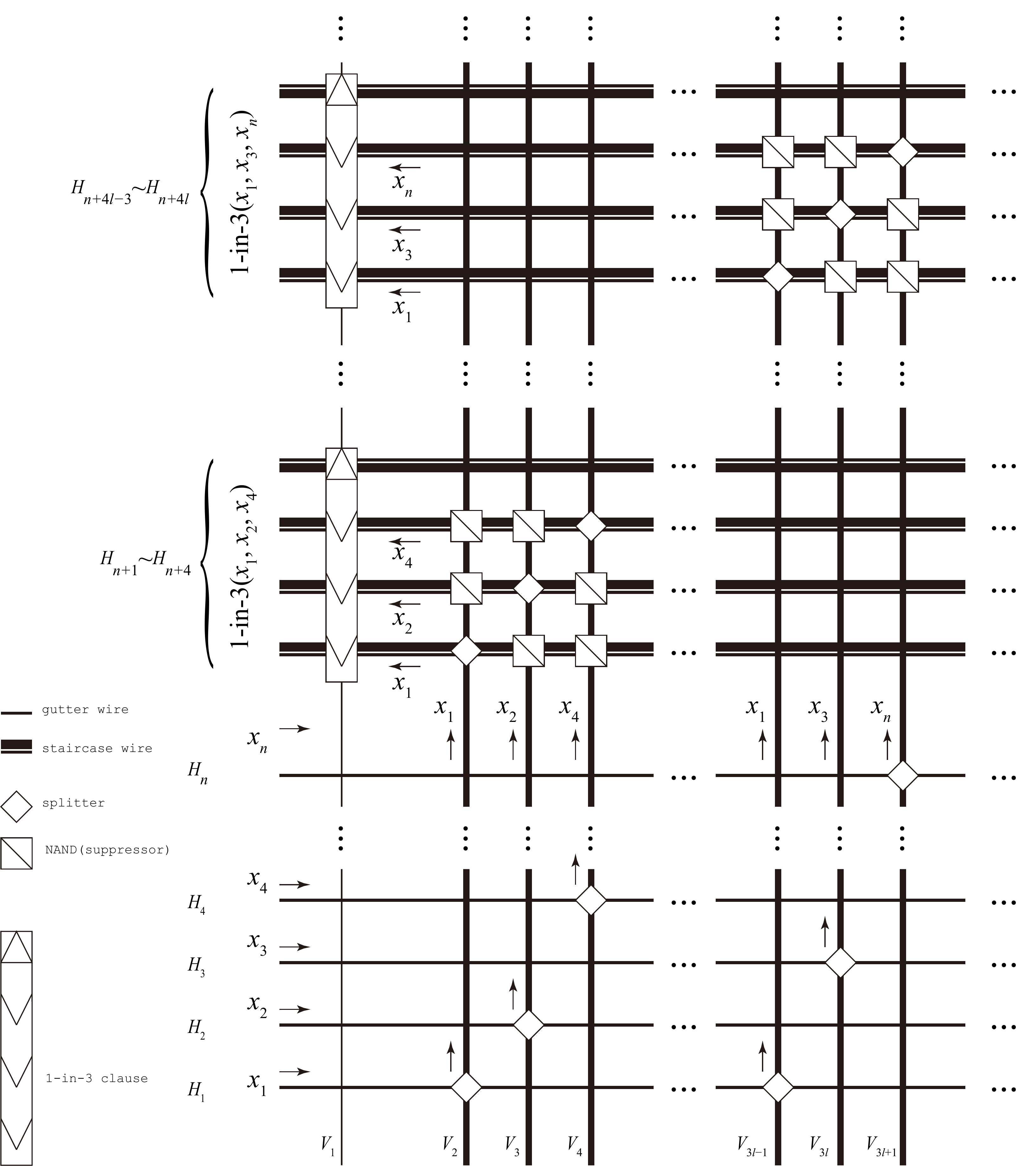}
	\centering
	\caption{Schematic diagram of gadget layout.}
	\label{fig:overview}
\end{figure}
To prove this, we reduce from 1-in-3 SAT. 
Figure~\ref{fig:reduction} shows the crease pattern for such a 1-in-3 SAT reduction, i.e., the crease pattern is rigidly foldable if and only if there is a set of binary variables that satisfy a given 1-in-3 SAT problem.
The basic idea is that we describe binary variables by whether creases are folded (true) or not (false).
This binary information will be carried via sets of four parallel creases, which we call \emph{wires}.
As we detail later, we carefully design \emph{gadgets}, i.e., crease patterns connected to wires that are designed to be rigidly foldable for a given binary pattern of wires. 

To overview the functional connection between gadgets, Figure~\ref{fig:overview} shows the schematic for designing such a crease pattern.
The gadgets are designed to have two inputs of horizontal (from the left) and vertical (from below) wires, or multiple inputs from different horizontal wires from the right (in the case of the 1-in-3 clause gadgets), so that the wires carry their signals through the gadgets. 
Specifically, we design a \emph{splitter gadget} (Figure~\ref{fig:splitter}) that copies a horizontal value to a vertical value and vice versa, a \emph{suppressor gadget} (Figure~\ref{fig:cross-over}) that enables a NAND operation between horizontal and vertical values, and a \emph{1-in-3 clause gadget} (Figure~\ref{fig:clause}) that rigidly folds if and only if exactly one of three horizontal inputs are true.
In addition, there are many intersections of horizontal and vertical wires. We carefully design our crease pattern so that these intersections merely cross over and do not logically interfere with one another.
The gadgets are arranged in the crease pattern as follows (more detail is given in Section~\ref{sec:proof}): 
\begin{enumerate}
\item The bottom $n$ horizontal wires represent the input variables $x_1, \dots x_n$.
        They are copied to vertical lines by using splitters and carried to each clause that uses this variable.
\item Each set of four horizontal wires above the input variables section correspond to a clause. 
        The copied variables in vertical lines are reflected horizontally using a $3\times3$ grid of gadgets composed of splitters and suppressors.
        This forces that at most $1$ of $3$ variables to be true, and the folding of 1-in-3 clause gadget (thus making the clause be true) will ensure that exactly one of the horizontal wires must be true.  
\end{enumerate}

To prove that this works, we follow the following strategy.
For each gadget, (1) we show that it successfully constrains the binary operations, i.e., it keeps wrong patterns from folding (the only if part), and (2) we give one MV assignment and rigid folding mode for each possible case of true and false (the if part).
The first part will make sure that only if $3$-SAT is solved, the pattern is rigidly foldable.
The second part shows that there exists at least one possible assignment that makes the whole thing rigidly foldable.

\subsection{The Gadgets}
\subsubsection{Wires}
We extract a set of four parallel lines from the tessellation structure, and call them \emph{true} if any of the four creases are folded, and \emph{false} if none of the four creases are folding.

For true values we will use $2$ rigid folding modes of wires, which we call \emph{gutter} and \emph{staircase}, respectively (Figure~\ref{fig:wire}).
Note that the folding speeds measured in tangents of half the fold angles are different between two pairs of creases for a staircase wire.
In our construction we will use a folding speed of $p=\frac{1}{2}$, and the relative speed of the pairs of creases in a staircase will be $p^2=\frac 1 4$.

\begin{figure}[t]
	\includegraphics[width=0.5\linewidth]{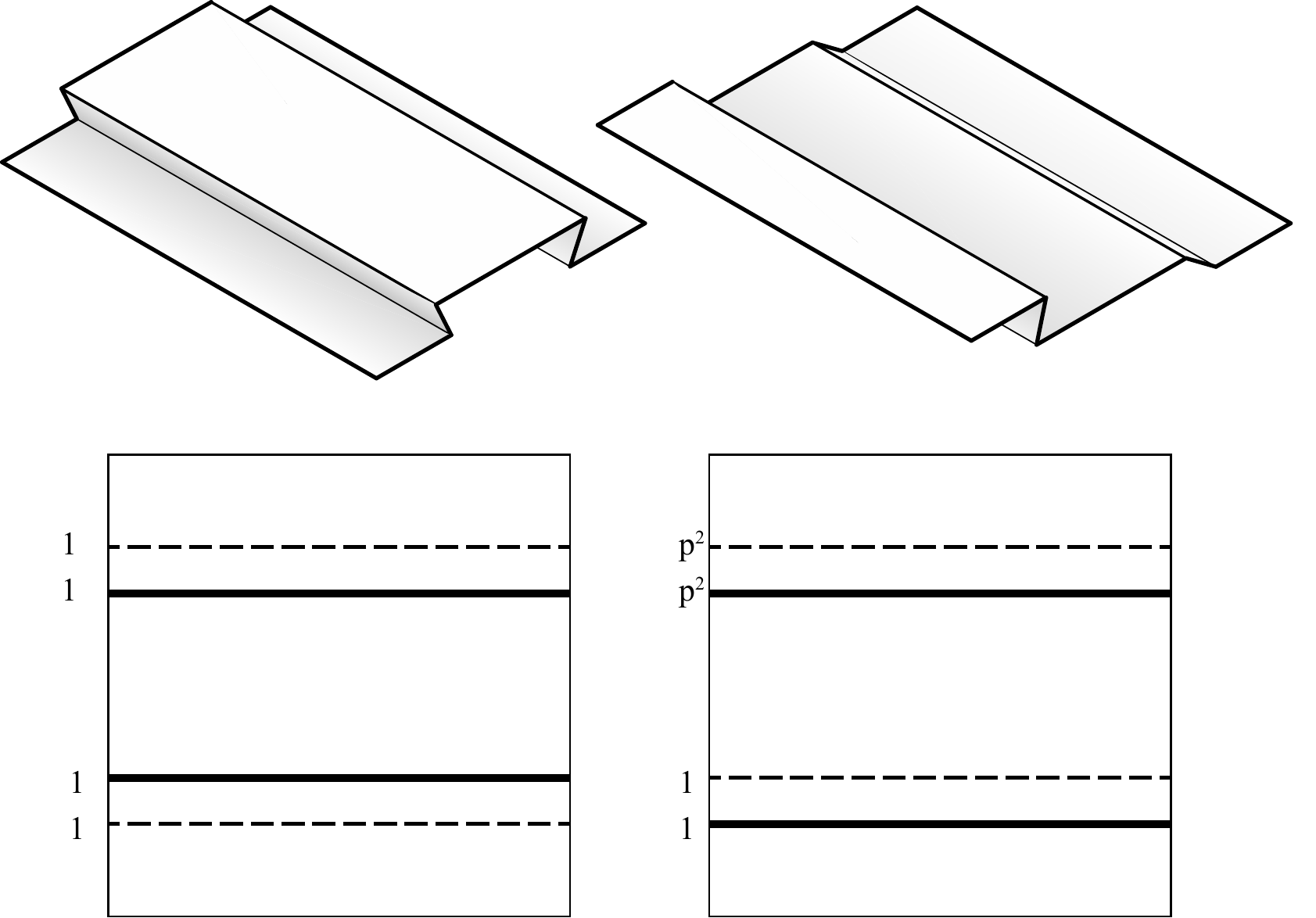}
	\centering
	\caption{Two types of wire gadgets in true setting. Left: gutter wire. Right: staircase wire. Note that the numbers show relative folding angle speed.}
	\label{fig:wire}
\end{figure}

\subsubsection{Splitters}
\begin{figure}[tbhp]
	\includegraphics[width=\linewidth]{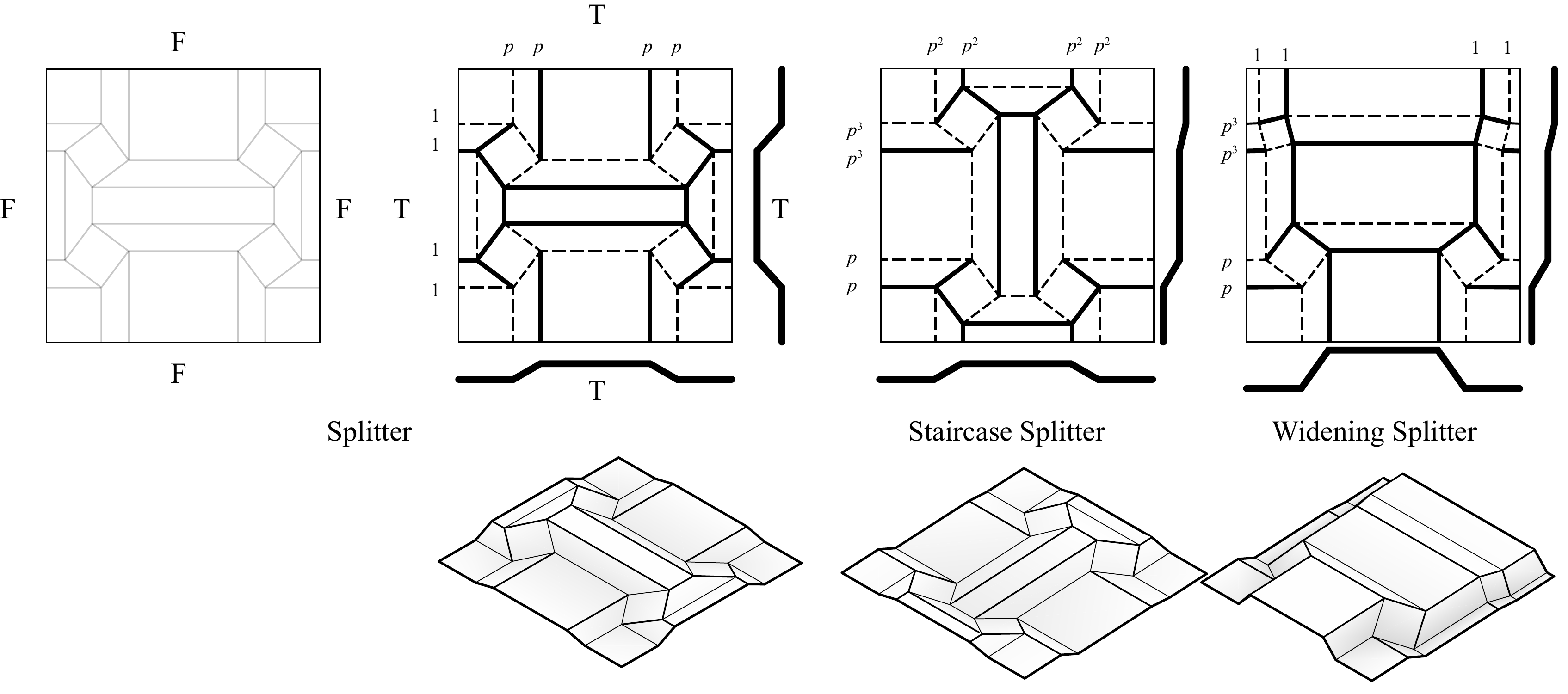}
	\centering
	\caption{Splitter gadget allowing for (true, true) or (false, false) only.}
	\label{fig:splitter}
\end{figure}

A \emph{splitter} is the $2\times 2$ twist fold tessellation with a twist angle of $\arctan{3\over 4}$, where the left and right horizontal wires are the input and the top and bottom vertical wires are the output (or the other way around) (Figure~\ref{fig:splitter} left two).  
Note that we may think of both the left and right wires to be inputs (and both top and bottom to be outputs) because signals will always be unchanged as they travel through the splitter, either horizontally or vertically.
A standard splitter will have all of its wires be gutters.
A \emph{staircase splitter} is a $90^\circ$ rotation of splitter, but its wires will have different modes so that the horizontal wire is a staircase and the vertical wire is a gutter.
A \emph{widening splitter} uses a pattern with two pairs of twists with angles $\arctan{3\over 4}$ (at the bottom) and $2\arctan {7\over9}$ (at the top) to change the widths of vertical wires.
The folding mode of a widening splitter is designed to match with staircase horizontal wires and gutter vertical wires.

In what follows we denote the binary values of our gadgets with the notation $(X, Y)$ where $X$ is the boolean value of the horizontal wires and $Y$ is the value of the vertical wires.

\begin{lemma}
\label{splitter-lemma}
A splitter, staircase splitter and widening splitter can fold only if they copy the values in the input to the output, and thus the possible patterns are (true, true) or (false, false).
\end{lemma}
\begin{proof}
All horizontal creases are connected to vertical creases via degree-$4$ vertices with finite value of speed coefficients.
If any of these creases fold, then every fold line in the gadget must fold, forcing the input and the output wires to have the same value.
\end{proof}

\begin{lemma}
\label{splitter-mode-lemma}
A splitter, staircase splitter and widening splitter can rigidly fold in the mode specified in Figure~\ref{fig:splitter}, where
the numbers $1, p, p^2, p^3$ ($p=\frac{1}{2}$) assigned to the fold lines are the absolute relative folding speeds measured in tangent of half the fold angles.
\end{lemma}
\begin{proof}
All splitters are composed of $4$ rigid origami twists as seen in Lemma~\ref{twist-mode}.
By consecutively mirroring the square twist in mode (a) of Figure~\ref{fig:twist} across horizontal and vertical axes, we obtain a splitter with the folding mode described in Figure~\ref{fig:splitter}.
Because $p(\arctan{3\over 4}) = \frac{1}{2}$, the folding follows the specified speed.

The staircase splitter is obtained by rotating the normal splitter, but then changing some of the modes of the vertices (refer to Figure~\ref{fig:single-vertex}) so as to cascade the relative folding speed $p$ as shown in Figure~\ref{fig:splitter}.
A widening splitter connects mode (a) of the square twist from Figure~\ref{fig:twist} with another square twist with a different twist angle.
The consistent folding mode is ensured by having  the bottom and top twists connected through folds with the same absolute speed and then having them mirrored with respect to the vertical line.
Now, because we chose $p(2\arctan {7\over9}) = \frac{1}{8}$ for the top twists,  the top two creases folds $p^2$ times the speed of the bottom two creases.
\end{proof}
\begin{remark}\label{rm:splitter-speed}
Note that vertical wires are all gutters.
Also, for the normal splitter, the speed of vertical wire is $p$ times the speed of horizontal wire.
Moreover, the speed of the vertical wire in the staircase splitter is $p^2$ times that of widening splitter (when we synchronize the horizontal staircase wires of staircase and widening splitters).
\end{remark}

\subsubsection{Suppressor and Crossover Gadgets}
Just overlaying two orthogonal wires (Figure~\ref{fig:cross-over} left) would make vertical and horizontal wires suppress each other so that they cannot be folded at the same time.
We call this a \emph{suppressor}, which only allows the patterns (false, false), (true, false), and (false, true).
\begin{lemma}
\label{suppressor-lemma}
A suppressor leaves all four horizontal creases unfolded or all four vertical creases unfolded.
\end{lemma}
\begin{proof}
Any horizontal crease crosses perpendicularly to all vertical creases, forming degree-$4$ vertex with $90^\circ$ sector angles.
Because the speed coefficients are $0$ at such a point, vertical creases cannot fold if horizontal crease have fold angle $(0, 180^\circ)$.
\end{proof}
This implies that we need to devise a crossover system to avoid logical interactions between horizontal and vertical wires.

A \emph{crossover gadget} (Figure~\ref{fig:cross-over}, right) is realized by overlaying the crease pattern of a splitter on top of a suppressor, where we line up the horizontal and vertical inputs.
The overlaying of patterns allows us to chose either pattern to be active, so we get the union of the folding modes of splitters and suppressor, which complete four possible binary patterns.
\begin{lemma}
\label{crossover-lemma}
Crossover gadgets can fold into four patterns (false, false), (true, false), (false, true) and (true, true).
The folding modes of (true, false) and (false, true) patterns follow those of the corresponding wires, and the folding mode of (true, true) pattern follows that of a splitter.
\end{lemma}
\begin{figure}[tbhp]
	\includegraphics[width=\linewidth]{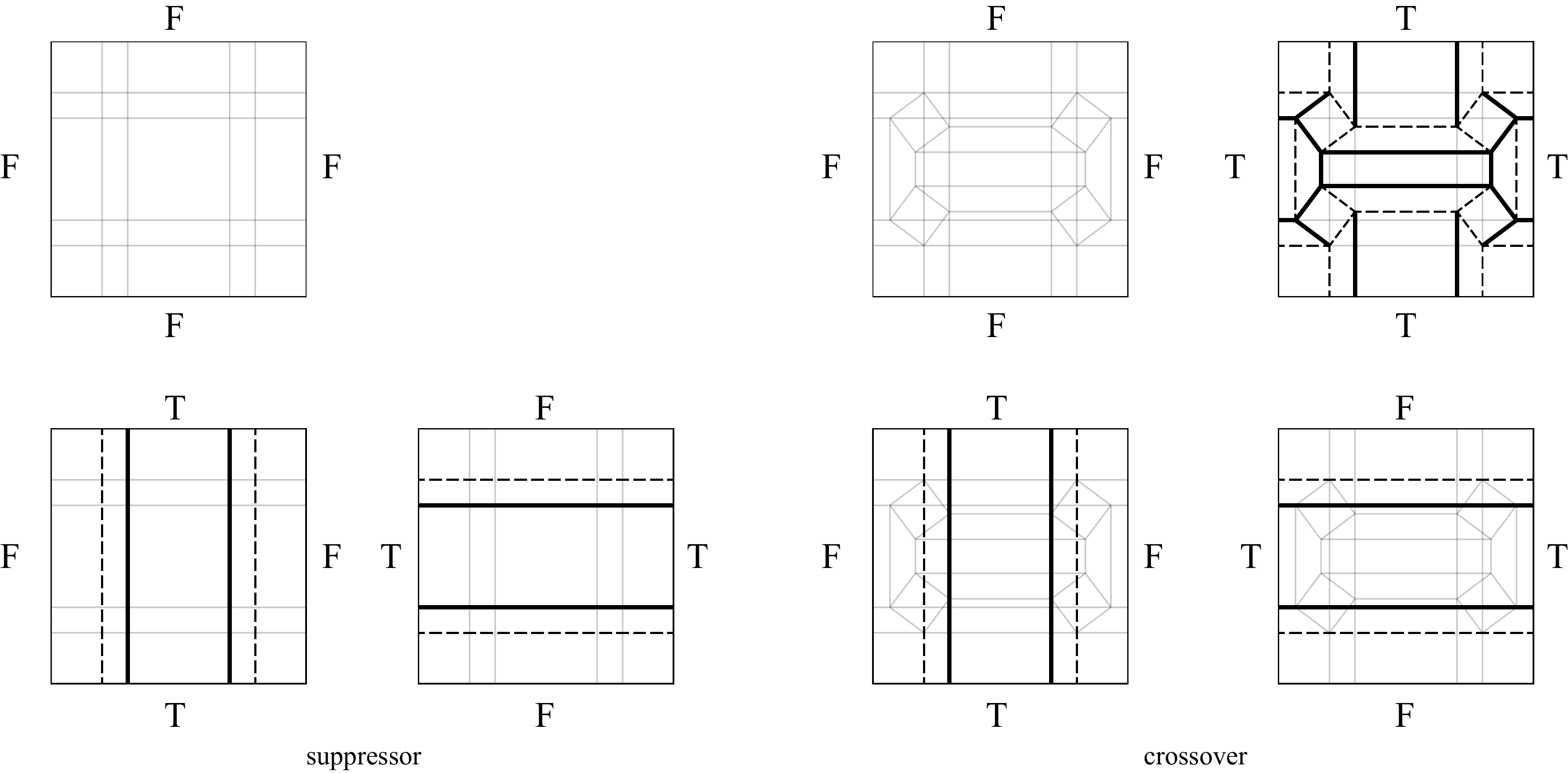}
	\centering
	\caption{Suppressor (left) can not achieve (true, true) pattern. Crossover gadget (right) can realize any combination (false, false) (false, true), (true, false), (true, true)}
	\label{fig:cross-over}
\end{figure}

\subsubsection{Clause Gadget}
\begin{figure}[t]
	\includegraphics[width=1.0\linewidth]{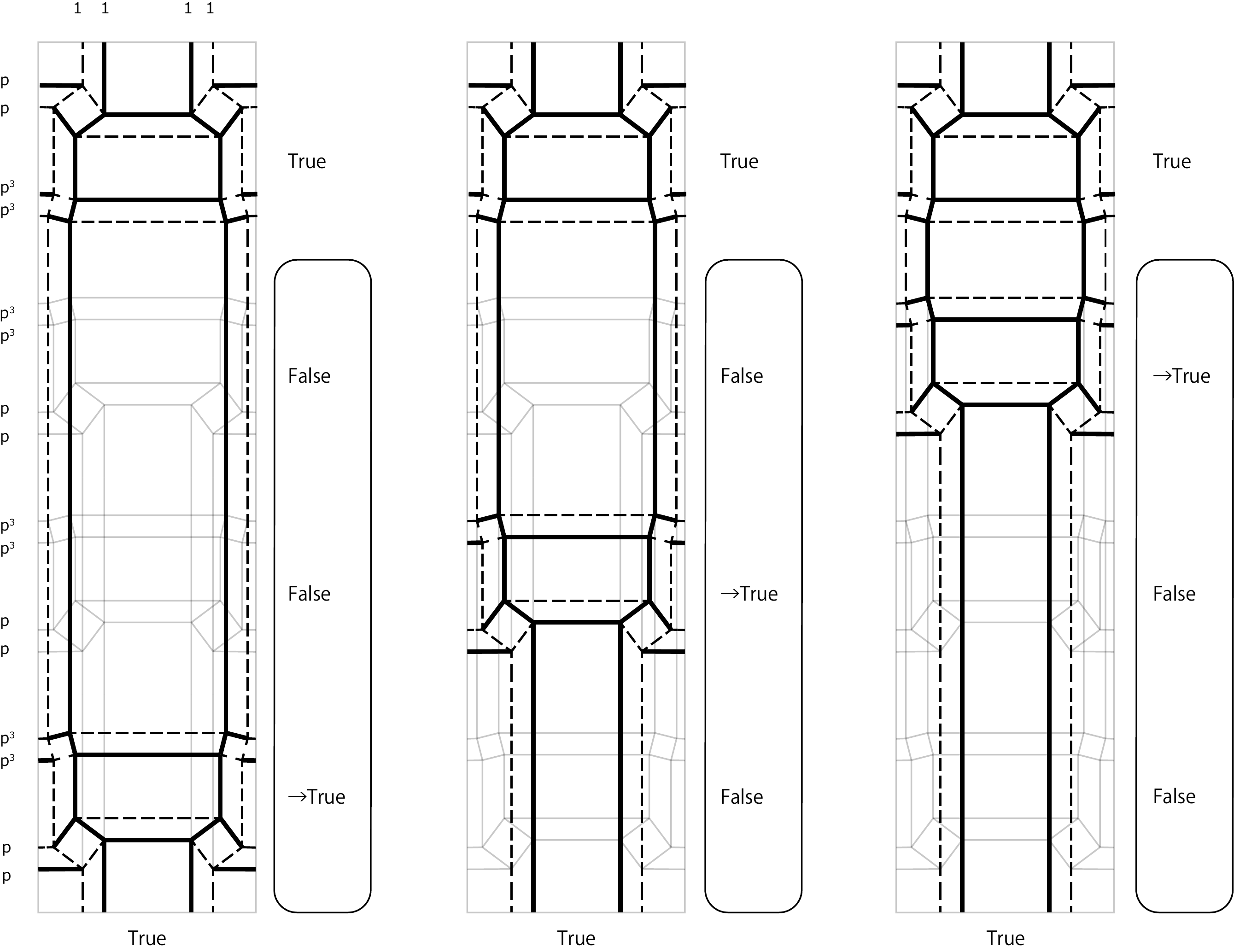}
	\centering
	\caption{Clause gadget. If vertical signal is true, this can fold into $3$ ways. Each of the folding mode corresponds to horizontal signals (true, false, false), (false, true, false), or (false, false, true). The gadgets forces us to select exactly one of the horizontal signals.}
	\label{fig:clause}
\end{figure}
A \emph{clause gadget} is designed by overlaying three widening splitters in vertical direction, connected by one upside-down widening splitter on top (see Figure~\ref{fig:clause}).
We consider three bottom horizontal wires to be the literals, and the vertical wire gives the satisfiability.
The fourth (top) horizontal wire, which is connected to the upside-down widening splitter, exists only for readjustment of widths to connect to the next clause on top.
\begin{lemma}
\label{clause-lemma}
A clause gadget folds only if exactly one of three inputs is true.
\end{lemma}
\begin{proof}
Assume that the second input wire (from the bottom) is true.
Because every horizontal crease in a clause gadget stops when it encounters a twist fold, all the twists must be folded in the second wire section. 
The vertical wires above the twists are wide, and the vertical wires bellow the twists are narrow.  Therefore the twists connected to the first and third input wires will be suppressed, and the clause must form a (false, true, false) pattern.
This is the same for the first and third horizontal wires, so we obtain (true, false, false) or (false, false, true) patterns.
Because the bottommost vertical wire is narrow and the topmost vertical wire is wide, the vertical wire is foldable only if it uses at least one of the twists, but doing so will force us to chose one of three modes we have described.
\end{proof}
\begin{lemma}\label{clause-mode-lemma}
Clause gadget can fold in $3$ modes specified in Figure~\ref{fig:clause}, which follows the folding mode of widening splitter and its upside down mirror image.
\end{lemma}
\begin{proof}
The rigid folding motion follows that of widening splitter, and because the connected upside down splitter is compatible, they fold rigidly with the specified relative speed (in tangent half fold angle parameterization) in Figure~\ref{fig:clause}.
\end{proof}

\subsection{The proof (Layout)}\label{sec:proof}
\begin{proof} [Proof of Theorem~\ref{optional}]
For a given 1-in-3 SAT conjunctive normal form $R_1 \wedge \cdots \wedge R_m$ with $m$ clauses $R_1, \dots, R_m$ of $n$ variables $x_1, \dots, x_n$,
we construct a grid of  $n+4m$ horizontal wires $H_1, \dots, H_{n+4m}$ and vertical wires $V_1,\dots,V_{3m+1}$,
where we place a gadget at each intersection node of wires or at groups of nodes as specified bellow.
By default,  crossover gadgets are placed at unspecified nodes between horizontal and vertical wires.

We set the horizontal wires from $H_1, \dots, H_{n}$ to be the variables $x_1, \dots, x_n$,
and each set of four horizontal wires $(H_{n+4(l-1)+1}, \dots,H_{n+4(l-1)+4})$ to be the $l$th clause $R_l$ ($l= 1, \dots, m$).
Specifically, for $R_l = \textrm{1-in-3}(x_i, x_j, x_k)$, we place a clause gadget at the intersection between $V_1$ and $(H_{n+4(l-1)+1},\dots,H_{n+4(l-1)+4})$.
The input variables $H_{n+4(l-1)+1},H_{n+4(l-1)+2},H_{n+4(l-1)+3}$ of the clause are copied from $V_{3l-1},V_{3l},V_{3l+1}$, respectively; and the values of $V_{3l-1},V_{3l},V_{3l+1}$ are copied from $H_i,H_j,H_k$, respectively, by placing splitters at appropriate the intersections.
We claim that the value of each of the vertical wires $V_{3l-1},V_{3l},V_{3l+1}$ is constrained only by the 1-in-3 clause gadgets that they are attached to via splitters.  This is because each of these vertical wires, if folding in a true signal, will (1) suppress the other vertical wires being fed into this clause and (2) cross other horizontal wires at a crossover gadget which will not change its value.
By Lemmas~\ref{splitter-lemma}, \ref{suppressor-lemma}, \ref{crossover-lemma} and~\ref{clause-lemma}, each clause can rigidly fold only if the clause returns true.

Now, because we concatenate all cause gadgets along $V_1$, all clauses must be true if any clause is true.
We claim that at least one clause must be true if any fold line is folded, therefore the whole laid out pattern can fold only if the given 1-in-3 SAT $R_1 \wedge \cdots \wedge R_m$ is satisfied.
The claim follows because if all the clauses are false, then all variables are false and thus all horizontal wires are false.
Because all vertical wires $V_{2}\dots V_{3m+1}$ are the copies of variables, all vertical wires are also false, and thus the paper remains unfolded.

If 1-in-3 SAT is satisfied, we may assign true and false to each wire as appropriate.
For wires with false, we can assume that they are removed from the crease pattern, which will result in a flat-foldable quadrivalent mesh composed of splitters, staircase splitters, widening splitters.
For these creases, we assign the specified folding modes in Lemmas~\ref{splitter-mode-lemma} and~\ref{clause-mode-lemma}, so that they are compatible through the gutter and staircase wires at correct speeds.
In particular, the synchronization between widening splitters and the staircase splitters is realized by the different speed between $V_1$ (with speed $1$) and $V_2, \dots V_{3m+1}$ (with speed $p^2$) as in Remark~\ref{rm:splitter-speed}.
To allow for a correct synchronization, we place $90^\circ$ rotated crossovers at the intersections between $V_1$ and $H_1, \dots, H_n$ so that the speed of the horizontal wire is $p$ times the speed of the vertical wire (Figure~\ref{fig:synchronization}).
More precisely, consider that the $l$th clause's $d$th variable ($d=0,1,2$) $x_i$ is true, then each cycle formed by (1) the splitter at $(V_{3l-1+d}, H_i)$, (2) the staircase splitter at $(V_{3l-1+x}, H_{n+4l-3+d})$, (3) the widening splitter at $(V_1,H_{n+4l-3+x})$ as a part of clause gadget, and (4) the $90^\circ$ rotated splitter at $(V_1,H_i)$ as a part of crossover gadget rigidly folds with synchronized speeds at the creases.
\end{proof}

\begin{figure}[tbhp]
	\includegraphics[width=\linewidth]{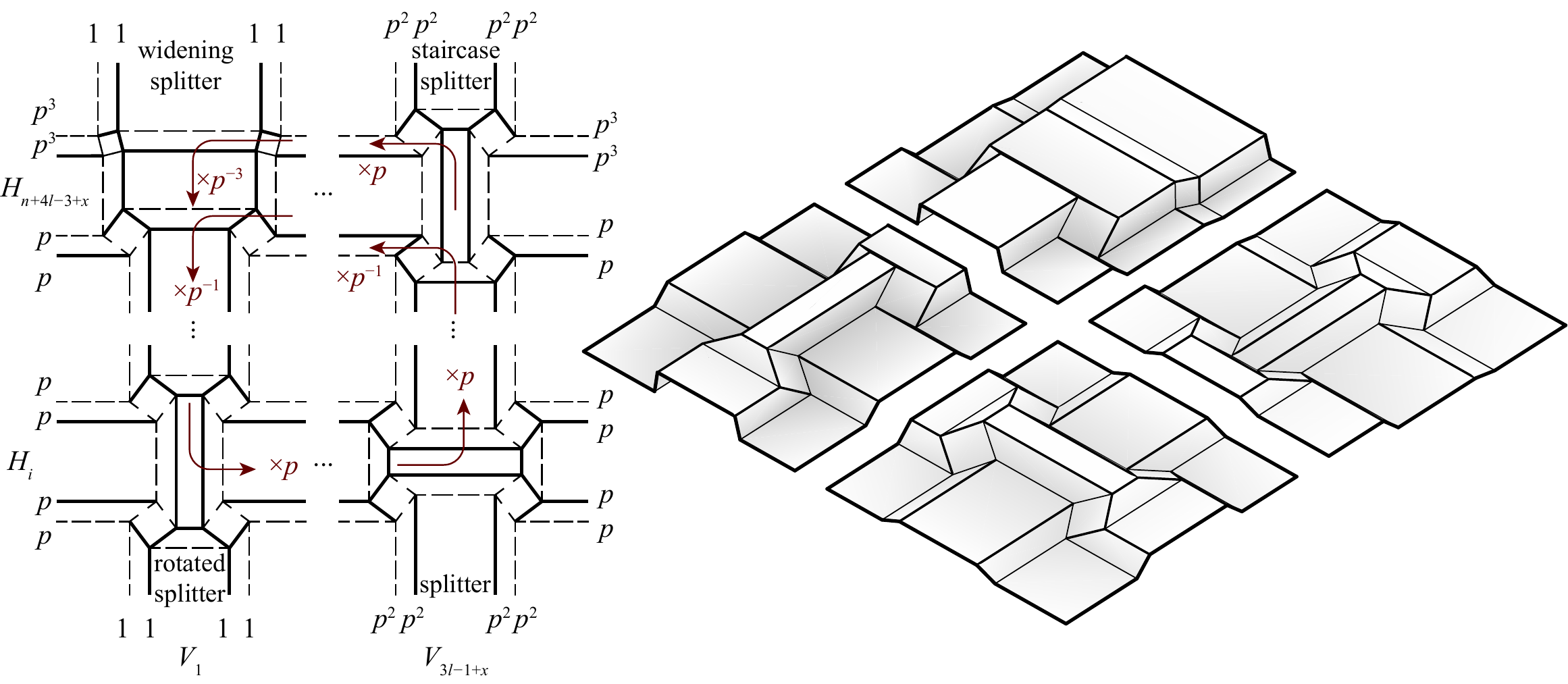}
	\centering
	\caption{Synchronization of gadgets: splitter (right bottom), staircase splitter (right top), widening splitter (left top), and $90^\circ$ rotated splitter (left bottom).  An example set of four gadgets are circle marked in Figure~\ref{fig:reduction}.}
	\label{fig:synchronization}
\end{figure}

\section{Conclusion}
We have shown that judging rigid foldability of a crease pattern using all creases is weakly NP-hard (Theorem~\ref{usingall}) and rigid foldability with optional creases is NP-hard (Theorem~\ref{optional}).
The lower bound is not proven to be tight, as we do not know if this problem belongs to NP in general.
As rigid origami is a special case of linkages, the rigid foldability problem belongs to Class $\exists \mathbb R$~\cite{AbelDDELS16}.
We would like to know a tight lower bound in the future studies.
For example, for flat-foldable degree-$4$-vertex origami, rigid foldability with given assignment is solvable in linear time, and thus rigid foldability belongs to NP (Theorem~\ref{deg4-np}).
Note that our work only proves weakly NP hardness of the model using all creases since the all-crease model does not allow for the overlay scheme for creating useful gadgets.
It is still open if the hardness of two models differ.

\section*{Acknowledgements}
This work was begun at the 2015 Bellairs Workshop on Computational Geometry, co-organized by Erik Demaine and Godfried Toussaint. 
We thank the other participants of the workshop for stimulating discussions.

\bibliographystyle{plain}       % APS-like style for physics
\bibliography{paper}   % name your BibTeX data base

\end{document}